\documentclass[conference]{IEEEtran}
\IEEEoverridecommandlockouts
\renewcommand{\baselinestretch}{.9875}
\usepackage{cite}
\usepackage{amsmath,amssymb,amsfonts,mathrsfs,amsthm}
\usepackage{algorithmic}
\usepackage{graphicx}
\usepackage{textcomp}
\usepackage{xcolor}
\usepackage{url}
\usepackage{adjustbox}
\usepackage{caption}
\usepackage{subcaption}
\newtheorem{theorem}{Theorem}
\newtheorem{remark}{Remark}
\def\BibTeX{{\rm B\kern-.05em{\sc i\kern-.025em b}\kern-.08em
    T\kern-.1667em\lower.7ex\hbox{E}\kern-.125emX}}
\DeclareMathOperator*{\argmax}{arg\,max}

\begin{document}

\title{Knowledge-Aware Semantic Communication System Design}

\author{\IEEEauthorblockN{Sachin Kadam and Dong In Kim}
\IEEEauthorblockA{{Department of Electrical and Computer Engineering} \\
{Sungkyunkwan University (SKKU), Suwon 16419, Republic of Korea}\\
Email: sachinkadam@skku.edu, dikim@skku.ac.kr}
\thanks{This research was supported in part by the Korean Government (MSIT) under the ICT Creative Consilience program (IITP-2020-0-01821) supervised by the IITP (Institute for Information \& Communications Technology Planning \& Evaluation).}
}

\maketitle

\begin{abstract}
The recent emergence of 6G raises the challenge of increasing the transmission data rate even further in order to break the barrier set by the Shannon limit. Traditional communication methods fall short of the 6G goals, paving the way for Semantic Communication (SemCom) systems. These systems find applications in wide range of fields such as economics, metaverse, autonomous transportation systems, healthcare, smart factories, etc. In SemCom systems, only the relevant information from the data, known as semantic data, is extracted to eliminate unwanted overheads in the raw data and then transmitted after encoding. In this paper, we first use the shared knowledge base to extract the keywords from the dataset. Then, we design an auto-encoder and auto-decoder that only transmit these keywords and, respectively, recover the data using the received keywords and the shared knowledge. We show analytically that the overall semantic distortion function has an upper bound, which is shown in the literature to converge. We numerically compute the accuracy of the reconstructed sentences at the receiver. Using simulations, we show that the proposed methods outperform a state-of-the-art method in terms of the average number of words per sentence.
\end{abstract}

\vspace{-.2cm}
\section{Introduction} \label{Sec:Intro}
\vspace{-.1cm}
As per the prediction in~\cite{rajatheva2020white}, semantic communication (SemCom) technology is identified as one of the key ingredients in 6G due to the requirement of low latency and high data rate transmissions. The recent emergence of SemCom technologies finds applications in wide range of fields such as economics~\cite{liew2022economics}, metaverse~\cite{ismail2022semantic}, autonomous transportation systems~\cite{yang2022semanticedge},  smart factories~\cite{luo2022semantic}, and so on.
In SemCom, we only transmit useful and necessary information to the recipients. The semantic extraction (SE) is a process wherein the useful and necessary features are extracted from the original raw data. For example, the essential speech features are extracted using an attention-based mechanism in~\cite{weng2021semantic,weng2021semantic2,tong2021federated}.  

During critical applications such as military operations, search operations by forest personnel in a dense forest, medical emergencies in remote areas, fire incidents in a remote agricultural land, the release of water from a nearby dam, etc., only the essential information needs to be communicated on an urgent basis. The messages could be in the form of text or audio and they come from a limited dataset. 
In a non-critical application, such as broadcasting a text/audio summary of commentary provided by live football commentators. Among all the words spoken by them, only a limited set of useful or important words are relevant to the game. These words are drawn from a limited dataset such as football vocabulary~\cite{footballvocab} which includes words such as \textit{goal, player names, red card, football, score, assist, half-time}, etc. This limited dataset provides an opportunity, in the context of SemCom design, for a significant overhead reduction by extracting and processing only the relevant keywords. For example, an uttered commentary sentence is:  `Ronaldo shoots the ball into the right-bottom of the net and it's a goal!'  The extracted keywords in this example are \textit{Ronaldo, shoots, ball, right-bottom, net, goal}. Only these keywords are transmitted in place of the entire sentence, and the receiver reconstructs a meaningful sentence. The reconstructed sentence in this case is: `Ronaldo shoots the ball into the right-bottom of the net to score a goal.' This sentence is not exactly the same as the original sentence, but it conveys the same meaning.

The main goal of this paper is to use SemCom technology to reduce communication overhead, in the context of natural language processing (NLP) problems, while maintaining a certain minimum accuracy in wireless communication systems. The overhead reduction is performed with high accuracy in the literature~\cite{xie2021deep,xie2020lite}. However, in some applications, high data rates are preferred over high accuracy. As a result, we present the results of the trade-off between overhead reduction and accuracy. Model parameters are chosen based on the context.
Instead of transmitting raw data, the transmitter is designed to transmit semantic data, which significantly reduces network data traffic. A knowledge base (KB) is a technology that collects, stores, and manages data. A knowledge graph (KG) is a KB that integrates data using a graph-structured topology. They are used to store interconnected event descriptions. These are used to predict the missing words in the received data (keywords) to construct a meaningful sentence.

The organization of the paper is as follows: A brief literature review on SemCom technologies is provided in Section~\ref{Sec:RelatedWork}. We introduce our proposed system model in Section~\ref{Sec:SysModel} and provide a few useful simulation results in Section~\ref{Sec:Simulations}. Finally, we conclude the paper in Section~\ref{Sec:Conclusions}.
\vspace{-.2cm}
\section{Related Work} \label{Sec:RelatedWork}
\vspace{-.2cm}
The following state-of-the-art survey papers provide in-depth discussions on various SemCom technologies and their applications~\cite{yang2022semantic,qin2021semantic,lan2021semantic}. Deep learning based SemCom technologies are proposed in ~\cite{xie2021deep,xie2020lite}. A brief tutorial on the
framework of SemCom and a method to calculate a bound on semantic data compression is provided in~\cite{niu2022towards}. The SemCom technology wherein both transmitter and receiver are empowered with the capability of contextual reasoning is proposed in~\cite{seo2021semantics}. The SemCom technology for a system where transmitter and receiver speak different languages is designed in~\cite{sana2022learning}. In~\cite{wang2022performance}, a SemCom framework for textual data transmission is proposed. In this framework, semantic information is represented by a KG made up of a set of semantic triples and the receiver recovers the original text using a graph-to-text generation model. 
All of these works focused on achieving an overhead reduction without compromising the accuracy of the received data. None of these works investigated the possibility of further overhead reduction, thereby improving transmission data rates, while sacrificing a little accuracy. This issue is addressed in this paper using a shared knowledge base.  

A significant research on the usage of KBs and KGs is carried out in the field of natural language processing (NLP). A survey paper based on KG is presented in~\cite{ji2021survey}. Similarly, another survey paper on KB text generation is presented in~\cite{yu2022survey}. A method to generate a summary of sentences by using a given set of keywords is proposed in~\cite{li2020keywords}. Similarly, a method to generate a summary of sentences by using a  knowledge base is shown in~\cite{huang2020knowledge}. Recently, KGs are utlized in the context of SemCom design~\cite{wang2022performance,zhou2022cognitive,liang2022life}. But these works do not focus on the issue presented in this paper, which is to design a SemCom system with a significant overhead reduction with a little compromise on accuracy.   
\section{System Model}\label{Sec:SysModel}
\begin{figure*}
\centering
\includegraphics[width=0.95\textwidth]{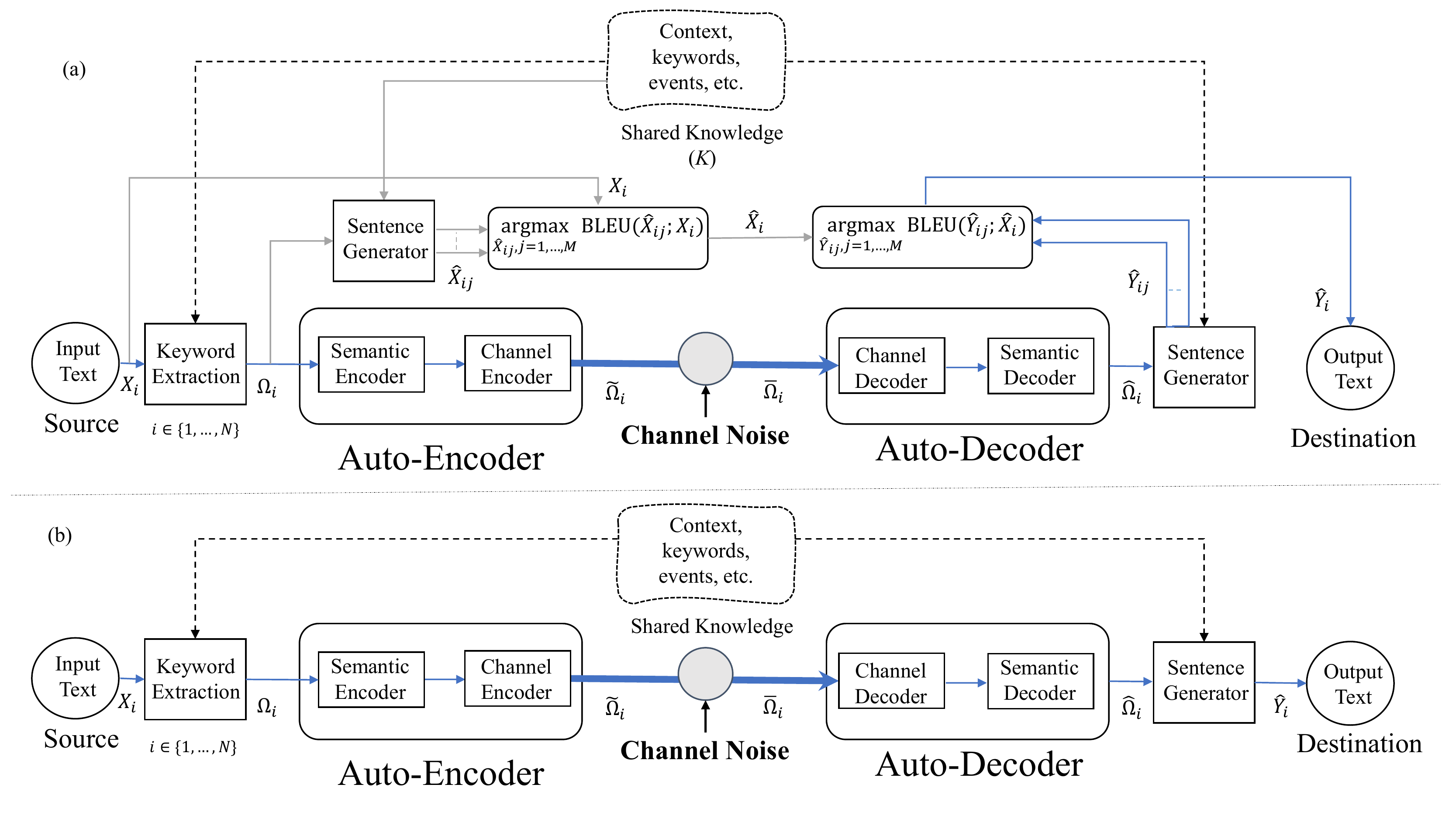}
    \caption{\small The block diagram of our proposed SemCom system model. The model in Fig. (a) is used for training the system parameters and the model in Fig. (b) is used for evaluating the system model.}
    \vspace{-.5cm}
\label{fig:SemComDesignModel}
\end{figure*} 
The system model of the proposed SemCom system is shown in Fig.~\ref{fig:SemComDesignModel}. Let $X$ be the input text dataset with $N$ sentences, $X_i$ be the $i^{th}, i\in \{1, \ldots, N\}$, sentence of $X$, and $K$ be the shared knowledge base (KB). First, we extract the keywords from $X$ using $K$. Let the total set of keywords be $\Omega = \bigcup_{i=1}^N \Omega_i$, where $\Omega_i$ denotes the set of keywords present in $X_i$. The keyword extraction process at every sentence $X_i, i\in \{1, \ldots, N\},$ is executed by multiplying it with a binary vector $b_i = [b_i(\ell), \ell = \{1, \ldots, |X_i|\}]$,\footnote{$|\mathcal{A}|$ denotes the cardinality of set $\mathcal{A}$.} which is defined as follows:
\begin{align}
b_i(\ell) &\triangleq {
\begin{cases} {1,}& {\text{if $\ell^{th}$ word of $X_i$, $X_i(\ell)$, is a keyword in $K$}} \\ 0,& \text{else}.
\end{cases}}
\end{align}
Hence, $\Omega_i$, $i\in \{1, \ldots, N\}$, is obtained by collecting the non-zero elements from $X_i \odot b_i$, where $\odot$ is a word-wise multiplication operator. Here $X_i \odot b_i \triangleq [X_i(\ell) b_i(\ell), \forall \ell = \{1, \ldots, |X_i|\}]$.\footnote{For ease of understanding, let us consider the example discussed in Section~\ref{Sec:Intro}. Let $X_i$ be `Ronaldo shoots the ball into the right-bottom of the net and it's a goal!'. If the set of keywords present in $X_i$ is \{\textit{Ronaldo, shoots, ball, right-bottom, net, goal}\} then $b_i = [1 1 0 1 0 0 1 0 0 1 0 0 0 1]$. Now, $X_i \odot b_i$ gives $[\textit{Ronaldo, shoots, 0, ball, 0, 0, right-bottom, 0, 0, net, 0, 0, 0, goal}]$. Next, $\Omega_i$ is obtained by collecting the non-zero elements, i.e., $\Omega_i = \{\textit{Ronaldo, shoots, ball, right-bottom, net, goal}\}$.}

Now, let us define the quantity BLEU score (bilingual evaluation understudy~\cite{papineni2002bleu}) to compare the similarities between two sentences quantitatively. 
The BLEU$(s,\hat{s}) \in [0,1]$ score between transmitted sentence $s$ and reconstructed sentence $\hat{s}$ is computed as follows:
\begin{equation}
    \text{BLEU}(s,\hat{s}) = \text{BP}(s,\hat{s}) \exp\left( \sum_{n=1}^{W} {w_n \ln{p_n (s,\hat{s})}}\right),
\end{equation}
where $p_n$ denotes the modified $n$-gram precision function up to length $W$, $w_n$ denotes the weights, and brevity penalty (BP) is given by the following expression:
\begin{align}
\text{BP}(s,\hat{s}) &={
\begin{cases} {1}& {\ell_c > \ell_r} \\ e^{1-\ell_r/\ell_c}& \ell_c \le \ell_r,
\end{cases}}
\end{align}
where $\ell_c$ is the length of the candidate translation and $\ell_r$ is the effective reference corpus length~\cite{papineni2002bleu}.

Let $\xi$ be a function which generates a set of $M$ (say) sentences from a given set of keywords with the help of a given knowledge base. Using the keywords in $\Omega_i, b_i$, $i\in \{1, \ldots, N\}$, and the knowledge $K$, for a given sentence $X_i$, the sentence generator at the transmitter generates a set of $M$ sentences using the function $\xi_\lambda$, where $\lambda$ is a parameter. Let that set of sentences be $\widehat{X}_{ij}, j=\{1, \ldots, M\}$. So,
\begin{equation}
    \widehat{X}_{ij} = \xi_\lambda(\Omega_i, K), ~j=\{1, \ldots, M\}.
\end{equation}
Next, out of these $M$ sentences we choose the most semantically equivalent sentence based on the BLEU scores~\cite{papineni2002bleu} compared with input sentence $X_i$, i.e., 
\begin{equation}
    \widehat{X}_i = \argmax_{\widehat{X}_{ij}, j=1, \ldots, M}\text{BLEU}(\widehat{X}_{ij}; X_i),~i\in \{1, \ldots, N\}.
\end{equation}
Note that the set of sentences $\widehat{X} = \{\widehat{X}_i, i\in \{1, \ldots, N\}\}$, is generated at the transmitter during the training process only and it is shared with the receiver a priori. Next, $i^{th}$ keyword set $\Omega_i$ is encoded using the auto-encoder which consists of semantic and channel encoders. The auto-encoder uses the binary vector $b_i, i\in \{1, \ldots, N\},$ to assign a common symbol to all non-keywords of $X_i$ and a unique symbol to keywords of $X_i$, respectively. This enables the receiver to construct appropriate sentences from the received symbol sets. Let us denote $\mathscr{S}_{\theta_e}$ and $\mathscr{C}_{\phi_e}$ as the semantic and channel encoders with $\theta_e$ and $\phi_e$ as the parameters sets, respectively. After encoding $\Omega_i$, we get the following set of symbols:
\begin{equation}
    \widetilde{\Omega}_i = \mathscr{C}_{\phi_e} (\mathscr{S}_{\theta_e}(\Omega_i)),~i\in \{1, \ldots, N\}. 
\end{equation}

The encoded set of symbols $\widetilde{\Omega}_i$ is transmitted via the AWGN (additive white Gaussion noise) channel. Let $h$ be the channel gain and $\eta$ be the noise which gets added to $\widetilde{\Omega}_i$ during transmission. So, the set of received symbols at the receiver is $\overline{\Omega}_i = h\widetilde{\Omega}_i + \eta$. After receiving, this set of symbols is decoded using the auto-decoder which consists of channel and semantic decoders. Let us denote $\mathscr{C}_{\phi_d}$ and $\mathscr{S}_{\theta_d}$ as the channel and semantic decoders with $\phi_d$ and $\theta_d$ as the parameters sets, respectively. After decoding $\overline{\Omega}_i$, we get the following set of keywords:
\begin{equation}
    \widehat{\Omega}_i = \mathscr{S}_{\theta_d} (\mathscr{C}_{\phi_d}(\overline{\Omega}_i)),~i\in \{1, \ldots, N\}.
\end{equation}

From the decoded set of keywords and the shared knowledge $K$, the sentence generator at the receiver generates a set of $M$ sentences using the function $\xi_\mu$, where $\mu$ is a parameter, and let that set of sentences be $\widehat{Y}_{ij}, j=\{1, \ldots, M\}$, i.e.,
\begin{equation}
    \widehat{Y}_{ij} = \xi_\mu(\widehat{\Omega}_i, K), ~j=\{1, \ldots, M\}.
\end{equation}
To select the most desired sentence among these sentences, we compute the BLEU scores between  $\widehat{Y}_{ij}, j=\{1, \ldots, M\}$ and the sentence at the transmitter $\widehat{X}_i$, and choose the one which maximizes the BLEU score. That is: 
\begin{equation}
\label{eq:LossMin}
    \widehat{Y}_i = \argmax_{\widehat{Y}_{ij}, j=1, \ldots, M}\text{BLEU}(\widehat{Y}_{ij}; \widehat{X}_i),~i=\{1, \ldots, N\},
\end{equation}
where $\widehat{Y} = \{\widehat{Y}_i, i\in \{1, \ldots, N\}\}$ denotes the desired set of sentences generated at the receiver.

\subsection{Training the System Model}
Let $\mathcal{X}$ be the set of all possible sentences and $p_\mathbb{X}(x), p_\lambda(\widehat{x}),$ and $p_\mu(\widehat{y})$ denote the probability distributions of sentences $\mathbb{X} \in X, \widehat{\mathbb{X}} \in \widehat{X},$ and $\widehat{\mathbb{Y}} \in \widehat{Y}$, respectively, for all $x,\widehat{x},\widehat{y} \in \mathcal{X}$.\footnote{Note that $\widehat{x},\widehat{y} \in \overline{\mathcal{X}} \subseteq \mathcal{X}$. To obtain the closed-form expressions for certain quantities, we relax the condition and assume $\overline{\mathcal{X}} = \mathcal{X}$.} The  sentences $\widehat{\mathbb{X}}$ and $\widehat{\mathbb{Y}}$ are generated by the sentence generators in transmitter and receiver, respectively, parameterized by $\lambda$ and $\mu$, respectively, with the help of shared knowledge $K$ (see Fig.~\ref{fig:SemComDesignModel}(a)). Note that, since $\widehat{\mathbb{X}}$ and $\widehat{\mathbb{Y}}$ are generated with the help of knowledge $K$, throughout the paper, the conditional distributions $p_\lambda(\cdot|\cdot)$ and $p_\mu(\cdot|\cdot)$ are conditioned on the event $\mathrm{K}=k$, $\forall k \in K$.

Let $H(\mathrm{X})$ and $\mathcal{D}_{KL}(p||q)$ represent, respectively, the entropy of a random variable $\mathrm{X}$ whose probability distribution is $p$ and the Kullback Leibler (KL) divergence between the probability distributions $p$ and $q$. These quantities are defined as follows~\cite{cover1999elements}:
    \begin{align}
        H(\mathrm{X}) &\triangleq -\sum_{x \in \mathcal{X}}  p(x) \log p(x), \label{eq:entropy}\\
        \mathcal{D}_{KL}(p||q) &\triangleq \sum_{x \in \mathcal{X}} p(x) \log \frac{p(x)}{q(x)}. \label{eq:KLDiv}
    \end{align}
The overall cross entropy (CE) loss measures the difference between the actual probability distribution at the input and
the estimated probability distribution at the output and it can be minimized using the stochastic gradient descent
(SGD) methods~\cite{yao2020negative}. So the overall cross entropy (CE) loss is defined as follows~\cite{goodfellow2016deep}:
\begin{align}
\label{eq:LossCE}
    \mathcal{L}^{CE}(\mu) \triangleq H(\mathbb{X}) + \mathcal{D}_{KL}(p_\mathbb{X}(x)||p_\mu(\widehat{y}|k)).
\end{align}
By using the expressions of~\eqref{eq:entropy} and~\eqref{eq:KLDiv}, we simplify the expression for overall CE loss as follows:
\begin{equation}
\label{eq:LossCE_New}
    \mathcal{L}^{CE}(\mu) = -\sum_{x \in \mathcal{X}}  p_\mathbb{X}(x) \log p_\mu(\widehat{y}|k).
\end{equation}

From Fig.~\ref{fig:SemComDesignModel}, we observe that there are information losses at auto-encoder, auto-decoder, and sentence generation blocks both in transmitter and receiver. We aim to minimize the summation of all these losses. The characterization of these losses are as follows.
\begin{itemize}

    \item The loss of information between sentences $\mathbb{X}$ and $\widehat{\mathbb{X}}$ at the transmitter is measured using the CE loss, i.e.,
    \begin{align}
        \mathcal{L}^{CE}_1 (\lambda) &= H(\mathbb{X}) + \mathcal{D}_{KL}(p_\mathbb{X}(x)||p_\lambda(\widehat{x}|k)) \\
        &= -\sum_{x \in \mathcal{X}} p_\mathbb{X}(x) \log p_\lambda(\widehat{x}|k). \label{eq:Loss1}
    \end{align}

    \item Similarly, the loss of information between sentences $\widehat{\mathbb{X}}$ and $\widehat{\mathbb{Y}}$ at the receiver is also measured using the CE loss, i.e.,
    \begin{align}
        \mathcal{L}^{CE}_2 (\mu,\lambda) &= H(\widehat{\mathbb{X}}|K) + \mathcal{D}_{KL}(p_\lambda(\widehat{x}|k)||p_\mu(\widehat{y}|k)) \\
        &= -\sum_{\widehat{x} \in \mathcal{X}} p_\lambda(\widehat{x}|k) \log p_\mu(\widehat{y}|k). \label{eq:Loss2}
    \end{align}

    \item Lastly, the loss of information in the channel is measured in terms of mutual information (MI) between transmitted symbols and received symbols, i.e.,
    \begin{equation}
        \mathcal{L}^{MI}_3 (\theta_e,\phi_e,\theta_d,\phi_d) = \text{I}(\widetilde{\Omega}; \overline{\Omega}). \label{eq:Loss3}
    \end{equation}
\end{itemize}
Now, we define the overall semantic distortion function as follows:\footnote{We use $(\cdot)$ in place of parameters, wherever convenient, for ease of representation.}
\begin{equation}
\label{eq:Loss_overall}
    \mathcal{L} (\theta_e,\phi_e,\theta_d,\phi_d,\lambda,\mu) \triangleq  \mathcal{L}^{CE}_1 (\cdot) +  \mathcal{L}^{CE}_2 (\cdot) -\gamma \mathcal{L}^{MI}_3 (\cdot),
\end{equation}
where $\gamma \ge 0$ is a hyper-parameter. In Theorem~\ref{Thm:CEloss},
we show that the overall semantic distortion $\mathcal{L}(\cdot)$ attains an upper bound which can be optimized. We aim to compute the optimal parameters of  auto-encoder, auto-decoder, and sentence generation blocks. These blocks are characterized by the parameters $\theta_e, \phi_e, \theta_d, \phi_d, \lambda,\mu$ and are obtained by training the system model, as shown in Fig.~\ref{fig:SemComDesignModel}(a), by minimizing the overall loss function defined in~\eqref{eq:Loss_overall}. 

\begin{theorem} \label{Thm:CEloss}
The overall semantic distortion function attains the following upper-bound:
    \begin{equation}
        \mathcal{L} (\theta_e,\phi_e,\theta_d,\phi_d,\lambda,\mu) \le \mathcal{L}^{CE}(\cdot) - \gamma \mathcal{L}^{MI}_3 (\cdot) = \mathrm{B}.
    \end{equation}
\end{theorem}
\begin{proof}
The proof is given in Appendix.
\end{proof}

\begin{remark}
The upper-bound $\mathrm{B}$, provided in Theorem~\ref{Thm:CEloss}, is proved to be optimized using the SGD algorithms~\cite{yao2020negative}. Also, in the published works~\cite{xie2021deep,sana2022learning}, the semantic distortion of the overall system is minimized using a similar expression as that of $\mathrm{B}$. 
\end{remark}
Hence, to minimize the overall semantic distortion defined in~\eqref{eq:Loss_overall}, we seek to minimize the upper-bound $\mathrm{B}$ provided in Theorem~\ref{Thm:CEloss}. The loss due to mutual information
$\mathcal{L}^{MI}_3 (\cdot)$ can be estimated using state-of-the-art mutual information
neural estimator (MINE)~\cite{belghazi2018mutual}.
\subsection{Accuracy versus Overhead Reduction Trade-off}
Given the limited size of knowledge base, though the accuracy of the reconstructed sentences in $\widehat{Y}$ may not be sufficiently high, the useful content in those sentences is summarized and conveyed to the receiver. This novel approach saves a significant amount of overhead. 

There exists a trade-off between overhead reduction and the accuracy that depends on the size of the knowledge base $K$. For example, if the set $K$ is small, only a few keywords are extracted, encoded, and transmitted from the given input sentences in $X$, implying a higher amount of average overhead reduction. On average, this results in a large amount of missing information, so the accuracy of the reconstructed sentences in $\widehat{Y}$ is expected to be low. 
On the other hand, if the set $K$ is large, a significant number of keywords are extracted from the given sentences in $X$, encoded, and transmitted, implying a lower average overhead reduction. On average, this results in a small amount of missing information, so the accuracy of the reconstructed sentences in $\widehat{Y}$ is expected to be high. This phenomenon is numerically shown in Section~\ref{Sec:Simulations}.

So, we aim at minimizing the transmission of average number of words per sentence (equivalent to maximizing the average overhead reduction) by keeping a certain minimum accuracy information $\tau$ in the received sentence, i.e.,
\begin{align}
    \min & ~ \frac{1}{N}\sum_{i=1}^N  |\Omega_i| \label{eq:min_words} \\
     &\text{BLEU}(\widehat{Y}_i;X_i) \ge \tau,~i=\{1, \ldots, N\}, \label{eq:tau_thr}
\end{align}
where $|\Omega_i|$ denotes the number of keywords in $\Omega_i$ that corresponds to sentence $X_i$.


\subsection{Shared Knowledge Base}
We generate the shared knowledge base $K$ by using the keywords from a limited dataset $\Omega$ which consists of only the relevant words of a particular event, like that of a football game in our case.  We assume that both the transmitter and receiver have access to $K$. It is shown in~\cite{bao2011towards} that the capacity of the channel can be increased beyond Shannon's limit by using a semantic encoder with low semantic ambiguity and a semantic decoder with strong inference ability and a large shared knowledge base. 
From Section~\ref{Sec:SysModel}, recall that in every sentence, only the words $w \in \Omega$ are uniquely encoded and transmitted to the receiver in their corresponding time slots. At other time slots, a common symbol is transmitted. By utilizing $K$, the receiver reconstructs the sentence based on the received symbols. To improve the accuracy of the reconstructed sentences, we can increase the size of $K$ by adding more keywords from the vocabulary generated using $X$. 
This result is shown using simulations in Section~\ref{Sec:Simulations}. 

\section{Simulation Results}\label{Sec:Simulations}
\begin{table}[]
\caption{Simulation parameters}
\vspace{-.2cm}
    \centering
    \begin{tabular}{|l|l|}
    \hline
        Number of matches used in training & 1580 \\ \hline
        Number of matches used in evaluation & 340 \\ \hline
        Number of epochs during training & 10 \\ \hline
        SNR & 6 dB \\
        \hline
        Learning rate & 0.001 \\
        \hline
        Batch Size & 64 \\
        \hline
        Channel & AWGN \\ \hline
    \end{tabular}
    \label{tab:sim_param}
    \vspace{-.4cm}
\end{table}
First, we evaluate the performance of the text data transmission in terms of accuracy using BLEU score~\cite{papineni2002bleu}.\footnote{We defined the BLEU score in Section~\ref{Sec:SysModel}.} In our work, we use the dataset provided in~\cite{zhang2021soccer}. We parse the football commentary data of 1920 matches from the website \url{goal.com}. The considered football matches are from Union of European Football Associations (UEFA) Champions League, UEFA Europa League, and Premier League between 2016 and 2020. 
The simulation parameters used for plots in this section are shown in Table~\ref{tab:sim_param}. The simulations are performed in a computer with NVIDIA GeForce RTX 3090 GPU and Intel Core i9-10980XE CPU with 256GB RAM. 

Let $\rho$ be the fraction of the total vocabulary $V$, which contains all the dataset words, to be added to $K$. $\rho=0$ indicates that no additional vocabulary is added and the system is evaluated  only with the initial keyword set $\Omega_0$. Based on the way of adding the vocabulary words to $\Omega_0$, we propose two types of schemes. In the first type, $\rho |V|$ vocabulary words are uniformly chosen at random from $V$ and added to $K$. In the second type, the words in $V$ are first arranged in the decreasing order of the frequency of appearances in the dataset, and then the first $\rho |V|$ vocabulary words are added to $K$. We call these schemes as `RANDOM' and `ORDERED', respectively. 

The accuracy performances of both the schemes and a deep learning based SemCom system method named DeepSC~\cite{xie2021deep}, in terms of BLEU score vs. $\rho$, are shown in Fig.~\ref{fig:BLEU}. From the plot we can infer that even with $\rho=0$, the initial keyword set can produce a BLEU score of 0.55 (for 1-gram). This shows that the context-related keywords produce good results. Also, we see that as we add more vocabulary words to $\Omega_0$, the BLEU score increases. 
For the same value of $\rho$ and $n$, the ORDERED scheme performs better than the RANDOM scheme because of the addition of high frequency words.
And, in terms of different $n$-grams, BLEU score decreases as $n$ increases, which is an expected result. In comparison to the DeepSC scheme, the proposed schemes perform poorly in terms of accuracy but outperform it in terms of overhead reduction, as shown below.  

Next, we evaluate the performance of the proposed schemes, in terms of the transmission of average number of words per sentence, with respect to DeepSC~\cite{xie2021deep} and the results are shown in Fig.~\ref{fig:w_bar_vs_rho}. Let $\overline{W}$ denote the average number of words per sentence. From the plot we observe that both the schemes outperform DeepSC. Among the proposed schemes, for a given $\rho$ the RANDOM scheme outperforms the ORDERED scheme. This is because, in the ORDERED scheme high frequency words are added which increases the number of words to be encoded in the input data as compared with the RANDOM scheme. 

Now, we solve the optimization problem presented in~\eqref{eq:min_words} and~\eqref{eq:tau_thr} using both the proposed schemes. For this purpose, we evaluate $\overline{W}$ vs. $\tau$ and the results are shown in Fig.~\ref{fig:w_bar_vs_tau}. From the plot we observe that both the schemes outperform DeepSC. Also, we see that the performance of both the schemes is same for a given accuracy threshold $\tau$. This is because, as shown in Fig.~\ref{fig:BLEU}, for a given value of $\rho \in (0,1)$, the ORDERED scheme outperforms the RANDOM scheme in terms of accuracy, whereas in Fig.~\ref{fig:w_bar_vs_rho}, the RANDOM scheme outperforms the ORDERED scheme in terms of overhead reduction. Hence, we can choose any one of the proposed methods to solve the optimization problem. 
\begin{figure}
\centering
\resizebox{0.9\columnwidth}{!}
{\includegraphics{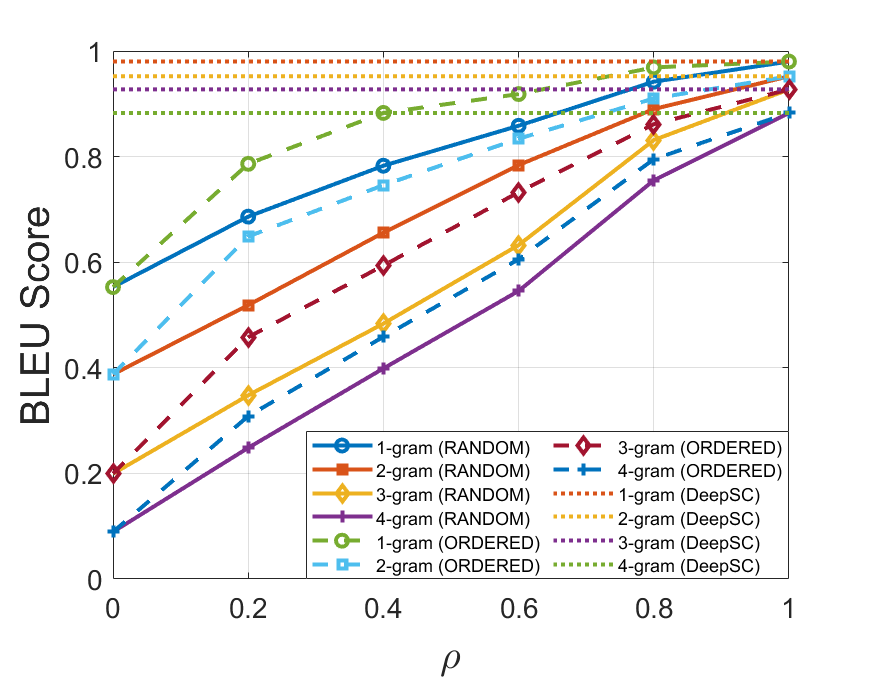}} 
\caption{\small This plot shows the BLEU score vs. $\rho$ for different values of $n$-grams, where $n=\{1,2,3,4\}$, for the proposed schemes and the DeepSC scheme~\cite{xie2021deep}.}
\vspace{-.4cm}
\label{fig:BLEU}
\end{figure}



\begin{figure}
\centering
\begin{subfigure}{.24\textwidth}
\centering
\begin{adjustbox}{width = 1\columnwidth}
\includegraphics[width=0.99\textwidth]{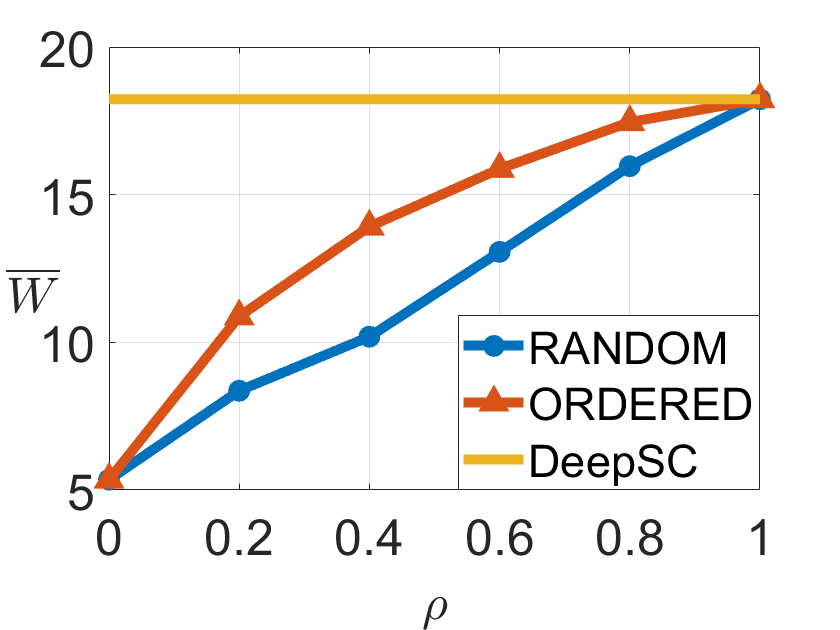}
\end{adjustbox}
\caption{}
\label{fig:w_bar_vs_rho}
\end{subfigure}%
\begin{subfigure}{.24\textwidth}
\centering
\begin{adjustbox}{width = 1\columnwidth}
\includegraphics[width=0.99\textwidth]{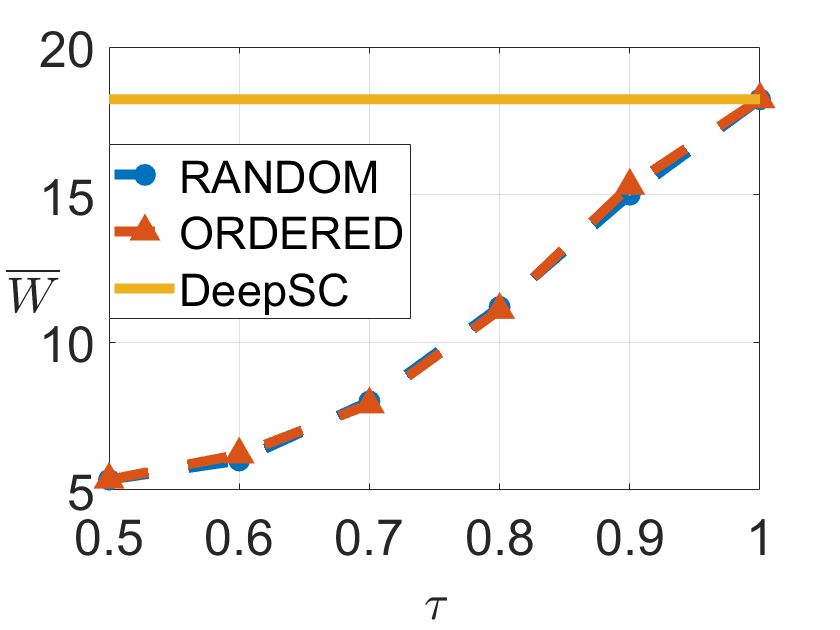}
\end{adjustbox}
\caption{}
\label{fig:w_bar_vs_tau}
\end{subfigure}
\vspace{-.2cm}
\caption{\small These plots show the average number of words per sentence vs. $\rho$ in the left plot and vs. $\tau$ in the right plot, respectively, for the proposed schemes and the DeepSC scheme~\cite{xie2021deep}.}
\vspace{-.4cm}
\label{fig:w_bar_plots}
\end{figure}

\section{Conclusions}\label{Sec:Conclusions}
In this paper, we first extracted relevant keywords from the dataset using the shared knowledge base. Then, using the received keywords and the shared knowledge, we designed an auto-encoder and auto-decoder that only transmit these keywords and, respectively, recover the data. We proved that the overall semantic distortion function has an upper bound, which is shown to be optimized using the SGD algorithms in the literature. We computed the accuracy of the reconstructed sentences at the receiver quantitatively. We demonstrated through simulations that the proposed methods outperform a state-of-the-art method in terms of average number of words per sentence. Furthermore, the proposed approach makes no new hardware modifications to the existing infrastructure. We focused solely on the text dataset; however, similar approaches can be used in the future for other types of datasets such as image, audio, and video. 

\appendix
Now we provide the proof of Theorem~\ref{Thm:CEloss}\label{Apdx:CEloss}. 
Let us define the following term, parameterized by $\lambda$, $\mu$, and $k \in K$:
\begin{equation}
\label{eq:delta_condition}
    \delta({\lambda,\mu,k}) \triangleq  \log \left(\frac{p_\mu(\widehat{y}|k)}{p_\lambda(\widehat{x}|k)}\right), ~\forall \widehat{x},\widehat{y} \in \mathcal{X}.
\end{equation}
From~\eqref{eq:Loss_overall}, we know that
\begin{subequations}
\begin{align}
    \mathcal{L}(\cdot) &= \mathcal{L}_1^{CE}(\cdot) +  \mathcal{L}_2^{CE}(\cdot) - \gamma \mathcal{L}^{MI}_3 (\cdot)\\
    & =-\sum_{x \in \mathcal{X}}  p_\mathbb{X}(x) \log {p_\lambda(\widehat{x}|k)} \nonumber \\
    &-\sum_{\widehat{x} \in \mathcal{X}} p_\lambda(\widehat{x}|k) \log p_\mu(\widehat{y}|k) - \gamma \text{I}(\widetilde{\Omega}; \overline{\Omega}) \label{eq:CEProof1}\\
    & =-\sum_{x \in \mathcal{X}}  p_\mathbb{X}(x) \log \left(p_\lambda(\widehat{x}|k)\frac{p_\mu(\widehat{y}|k)}{p_\mu(\widehat{y}|k)}\right) \nonumber\\
    &-\sum_{\widehat{x} \in \mathcal{X}} p_\lambda(\widehat{x}|k) \log \left(p_\mu(\widehat{y}|k)\frac{p_\lambda(\widehat{x}|k)}{p_\lambda(\widehat{x}|k)}\right) - \gamma \text{I}(\widetilde{\Omega}; \overline{\Omega}) \label{eq:CEProof2}\\
    &= -\sum_{x \in \mathcal{X}}  p_\mathbb{X}(x) \log p_\mu(\widehat{y}|k) + \delta({\lambda,\mu,k})  \nonumber \\
    &-\sum_{\widehat{x} \in \mathcal{X}} p_\lambda(\widehat{x}|k) \log p_\lambda(\widehat{x}|k) - \delta({\lambda,\mu,k})  - \gamma \text{I}(\widetilde{\Omega}; \overline{\Omega}) \label{eq:CEProof3}\\
    &=\mathcal{L}^{CE}(\cdot) + H(\widehat{\mathbb{X}}|K) - \gamma \text{I}(\widetilde{\Omega}; \overline{\Omega}) \label{eq:CEProof4}\\
    & \le \mathcal{L}^{CE}(\cdot) - \gamma \text{I}(\widetilde{\Omega}; \overline{\Omega}). \label{eq:CEProof5}
\end{align}
\end{subequations}
In~\eqref{eq:CEProof1}, we expand the loss function expressions using their respective definitions provided in~\eqref{eq:Loss1},~\eqref{eq:Loss2}, and~\eqref{eq:Loss3}, respectively. In~\eqref{eq:CEProof2}, we multiply and divide $p_\mu(\widehat{y}|k)$ and $p_\lambda(\widehat{x}|k)$ in the first and second terms, respectively. Using~\eqref{eq:delta_condition} and algebraic simplifications, we get~\eqref{eq:CEProof3}. By using~\eqref{eq:LossCE} and the definition of entropy, we write~\eqref{eq:CEProof4}. And, finally the inequality in~\eqref{eq:CEProof5} is due to $H(\widehat{\mathbb{X}}|K) \ge 0$~\cite{cover1999elements}.
\bibliographystyle{ieeetr}
\bibliography{references}

\begin{thebibliography}{10}

\bibitem{rajatheva2020white}
N.~Rajatheva, I.~Atzeni, E.~Bj{\"o}rnson, A.~Bourdoux, S.~Buzzi, J.-B.
  Dor{\'e}, S.~Erkucuk, M.~Fuentes, K.~Guan, Y.~Hu, {\em et~al.}, ``{White
  paper on broadband connectivity in 6G},'' {\em 6G Research Visions}, vol.~10,
  2020.

\bibitem{liew2022economics}
Z.~Q. Liew, H.~Du, W.~Y.~B. Lim, Z.~Xiong, D.~Niyato, C.~Miao, and D.~I. Kim,
  ``{Economics of Semantic Communication System: An Auction Approach},'' {\em
  arXiv preprint arXiv:2208.05040}, 2022.

\bibitem{ismail2022semantic}
L.~Ismail, D.~Niyato, S.~Sun, D.~I. Kim, M.~Erol-Kantarci, and C.~Miao,
  ``{Semantic Information Market For The Metaverse: An Auction Based
  Approach},'' {\em arXiv preprint arXiv:2204.04878}, 2022.

\bibitem{yang2022semanticedge}
W.~Yang, Z.~Q. Liew, W.~Y.~B. Lim, Z.~Xiong, D.~Niyato, X.~Chi, X.~Cao, and
  K.~B. Letaief, ``Semantic communication meets edge intelligence,'' {\em arXiv
  preprint arXiv:2202.06471}, 2022.

\bibitem{luo2022semantic}
X.~Luo, H.-H. Chen, and Q.~Guo, ``Semantic communications: Overview, open
  issues, and future research directions,'' {\em IEEE Wireless Communications},
  2022.

\bibitem{weng2021semantic}
Z.~Weng and Z.~Qin, ``Semantic communication systems for speech transmission,''
  {\em IEEE Journal on Selected Areas in Communications}, vol.~39, no.~8,
  pp.~2434--2444, 2021.

\bibitem{weng2021semantic2}
Z.~Weng, Z.~Qin, and G.~Y. Li, ``Semantic communications for speech signals,''
  in {\em ICC 2021-IEEE International Conference on Communications}, pp.~1--6,
  IEEE, 2021.

\bibitem{tong2021federated}
H.~Tong, Z.~Yang, S.~Wang, Y.~Hu, W.~Saad, and C.~Yin, ``Federated learning
  based audio semantic communication over wireless networks,'' in {\em 2021
  IEEE Global Communications Conference (GLOBECOM)}, pp.~1--6, IEEE, 2021.

\bibitem{footballvocab}
``Football/ soccer english vocabulary,
  https://www.vocabulary.cl/english/football-soccer.htm.''

\bibitem{xie2021deep}
H.~Xie, Z.~Qin, G.~Y. Li, and B.-H. Juang, ``Deep learning enabled semantic
  communication systems,'' {\em IEEE Transactions on Signal Processing},
  vol.~69, pp.~2663--2675, 2021.

\bibitem{xie2020lite}
H.~Xie and Z.~Qin, ``A lite distributed semantic communication system for
  internet of things,'' {\em IEEE Journal on Selected Areas in Communications},
  vol.~39, no.~1, pp.~142--153, 2020.

\bibitem{yang2022semantic}
W.~Yang, H.~Du, Z.~Liew, W.~Y.~B. Lim, Z.~Xiong, D.~Niyato, X.~Chi, X.~S. Shen,
  and C.~Miao, ``{Semantic communications for 6G future internet: Fundamentals,
  applications, and challenges},'' {\em arXiv preprint arXiv:2207.00427}, 2022.

\bibitem{qin2021semantic}
Z.~Qin, X.~Tao, J.~Lu, and G.~Y. Li, ``{Semantic communications: Principles and
  challenges},'' {\em arXiv preprint arXiv:2201.01389}, 2021.

\bibitem{lan2021semantic}
Q.~Lan, D.~Wen, Z.~Zhang, Q.~Zeng, X.~Chen, P.~Popovski, and K.~Huang, ``{What
  is semantic communication? A view on conveying meaning in the era of machine
  intelligence},'' {\em Journal of Communications and Information Networks},
  vol.~6, no.~4, pp.~336--371, 2021.

\bibitem{niu2022towards}
K.~Niu, J.~Dai, S.~Yao, S.~Wang, Z.~Si, X.~Qin, and P.~Zhang, ``{Towards
  Semantic Communications: A Paradigm Shift},'' {\em arXiv preprint
  arXiv:2203.06692}, 2022.

\bibitem{seo2021semantics}
H.~Seo, J.~Park, M.~Bennis, and M.~Debbah, ``Semantics-native communication
  with contextual reasoning,'' {\em arXiv preprint arXiv:2108.05681}, 2021.

\bibitem{sana2022learning}
M.~Sana and E.~C. Strinati, ``{Learning semantics: An opportunity for effective
  6G communications},'' in {\em 2022 IEEE 19th Annual Consumer Communications
  \& Networking Conference (CCNC)}, pp.~631--636, IEEE, 2022.

\bibitem{wang2022performance}
Y.~Wang, M.~Chen, T.~Luo, W.~Saad, D.~Niyato, H.~V. Poor, and S.~Cui,
  ``Performance optimization for semantic communications: An attention-based
  reinforcement learning approach,'' {\em IEEE Journal on Selected Areas in
  Communications}, vol.~40, no.~9, pp.~2598--2613, 2022.

\bibitem{ji2021survey}
S.~Ji, S.~Pan, E.~Cambria, P.~Marttinen, and S.~Y. Philip, ``A survey on
  knowledge graphs: Representation, acquisition, and applications,'' {\em IEEE
  Transactions on Neural Networks and Learning Systems}, vol.~33, no.~2,
  pp.~494--514, 2021.

\bibitem{yu2022survey}
W.~Yu, C.~Zhu, Z.~Li, Z.~Hu, Q.~Wang, H.~Ji, and M.~Jiang, ``A survey of
  knowledge-enhanced text generation,'' {\em ACM Computing Surveys (CSUR)},
  2022.

\bibitem{li2020keywords}
H.~Li, J.~Zhu, J.~Zhang, C.~Zong, and X.~He, ``Keywords-guided abstractive
  sentence summarization,'' {\em Proceedings of the AAAI conference on
  artificial intelligence}, vol.~34, no.~05, pp.~8196--8203, 2020.

\bibitem{huang2020knowledge}
L.~Huang, L.~Wu, and L.~Wang, ``{Knowledge Graph-Augmented Abstractive
  Summarization with Semantic-Driven Cloze Reward},'' in {\em Proceedings of
  the 58th Annual Meeting of the Association for Computational Linguistics},
  pp.~5094--5107, 2020.

\bibitem{zhou2022cognitive}
F.~Zhou, Y.~Li, X.~Zhang, Q.~Wu, X.~Lei, and R.~Q. Hu, ``Cognitive semantic
  communication systems driven by knowledge graph,'' {\em arXiv preprint
  arXiv:2202.11958}, 2022.

\bibitem{liang2022life}
J.~Liang, Y.~Xiao, Y.~Li, G.~Shi, and M.~Bennis, ``Life-long learning for
  reasoning-based semantic communication,'' {\em arXiv preprint
  arXiv:2202.01952}, 2022.

\bibitem{papineni2002bleu}
K.~Papineni, S.~Roukos, T.~Ward, and W.-J. Zhu, ``{BLEU: A method for automatic
  evaluation of machine translation},'' in {\em Proceedings of the 40th annual
  meeting of the Association for Computational Linguistics}, pp.~311--318,
  2002.

\bibitem{cover1999elements}
T.~M. Cover, {\em Elements of information theory}.
\newblock John Wiley \& Sons, 1999.

\bibitem{yao2020negative}
H.~Yao, D.-l. Zhu, B.~Jiang, and P.~Yu, ``Negative log likelihood ratio loss
  for deep neural network classification,'' in {\em Proceedings of the Future
  Technologies Conference (FTC) 2019: Volume 1}, pp.~276--282, Springer, 2020.

\bibitem{goodfellow2016deep}
I.~Goodfellow, Y.~Bengio, and A.~Courville, {\em Deep learning}.
\newblock MIT press, 2016.

\bibitem{belghazi2018mutual}
M.~I. Belghazi, A.~Baratin, S.~Rajeshwar, S.~Ozair, Y.~Bengio, A.~Courville,
  and D.~Hjelm, ``{Mutual Information Neural Estimation},'' in {\em
  International Conference on Machine Learning}, pp.~531--540, PMLR, 2018.

\bibitem{bao2011towards}
J.~Bao, P.~Basu, M.~Dean, C.~Partridge, A.~Swami, W.~Leland, and J.~A. Hendler,
  ``Towards a theory of semantic communication,'' in {\em 2011 IEEE Network
  Science Workshop}, pp.~110--117, IEEE, 2011.

\bibitem{zhang2021soccer}
R.~Zhang and C.~Eickhoff, ``{SOCCER: An Information-Sparse Discourse State
  Tracking Collection in the Sports Commentary Domain},'' in {\em Proceedings
  of the 2021 Conference of the North American Chapter of the Association for
  Computational Linguistics: Human Language Technologies}, pp.~4325--4333,
  2021.

\end{thebibliography}

\end{document}